\documentclass[11pt,onecolumn]{IEEEtran}

\usepackage[a4paper,hmargin=1in,vmargin=1.16in]{geometry}

\usepackage{color}

\usepackage{amssymb,amsmath,amsthm}
\usepackage[T1]{fontenc}
\usepackage[mathscr]{euscript}
\usepackage{microtype}
\usepackage{cite}

\newtheorem{lemma}{Lemma}
\newtheorem{remark}{Remark}

\newtheorem{proposition}[lemma]{Proposition}

\newtheorem{theorem}[lemma]{Theorem}

\newcommand{\cX}{\mathcal{X}}
\newcommand{\cY}{\mathcal{Y}}
\newcommand{\cZ}{\mathcal{Z}}
\newcommand{\cM}{\mathcal{M}}
\newcommand{\cP}{\mathcal{P}}
\newcommand{\cC}{\mathcal{C}}
\newcommand{\cI}{\mathcal{I}}

\newcommand{\perr}{p_{\textrm{err}}}
\newcommand{\abs}[1]{|#1|}
\newcommand{\eps}{\varepsilon}
\newcommand{\proj}[1]{|k\rangle\!\langle k|}
\newcommand{\cp}{\!\times\!}
\newcommand{\twonorm}[1]{\| #1 \|_{2}}
\newcommand{\twonormb}[1]{\big\| #1 \big\|_{2}}

\renewcommand{\vec}[1]{\mathbf{#1}}

\DeclareMathOperator{\Exp}{E}

\DeclareMathOperator*{\argmin}{arg\,min}

\newcommand{\PPV}{Polyanskiy \emph{et al.}}

\begin{document}
\flushbottom


\title{A Tight Upper Bound for the Third-Order Asymptotics for Most Discrete Memoryless Channels} 

\author{Marco Tomamichel,~\IEEEmembership{Member,~IEEE} and
        Vincent Y.~F.~Tan,~\IEEEmembership{Member,~IEEE}%
\thanks{M.~Tomamichel is
with the Centre of Quantum Technologies,
National University of Singapore, 3 Science Drive 2, Singapore 117542.
(email:\,cqtmarco@nus.edu.sg).
V.~Y.~F.~Tan is with the Institute for Infocomm Research (I$^2$R), A*STAR, Singapore and with
the Department of Electrical and Computer Engineering, National University of Singapore.
(email:\,vtan@nus.edu.sg).
This paper was presented in part at the IEEE International Symposium on Information Theory (ISIT 2013).}}


\IEEEpeerreviewmaketitle



\maketitle

\begin{abstract}
This paper shows that the logarithm of the $\eps$-error capacity (average error probability) for $n$ uses of a discrete memoryless channel (DMC) is upper bounded by the normal approximation plus a third-order  term that does not exceed $\frac{1}{2} \log n +O(1)$
if the $\eps$-dispersion of the channel is positive. This matches a lower bound by Y.~Polyanskiy (2010)  for DMCs with positive reverse dispersion.  If the $\eps$-dispersion vanishes, the logarithm of the  $\eps$-error capacity  is upper bounded by $n$ times the capacity plus a constant term except for a small class of DMCs and $\eps \ge \frac{1}{2}$.
\end{abstract}

%

\section{Introduction}


\label{sec:intro}
The primary information-theoretic task in point-to-point channel coding is the characterization of the maximum rate of communication over $n$ independent  uses of  a noisy   channel $W$. We are concerned in this paper with {\em discrete memoryless channels} (DMCs).  Let $M^*(W^n,\eps)$ resp.\  $M_{\max}^*(W^n,\eps)$  denote the maximum size of a length-$n$  block code for DMC $W$    having {\em average}   resp.\ {\em maximal}  error probability no larger than $\eps \in (0,1)$.  Shannon's {\em noisy-channel coding theorem} \cite{Shannon48} and   Wolfowitz's {\em strong converse}~\cite{Wolfowitz}   state that  for every $\eps \in (0,1)$,
\begin{equation}
\lim_{n\to\infty}\frac{1}{n}\log M^*(W^n,\eps)=C \quad\mbox{bits/channel use},\nonumber
\end{equation}
where $C := \max_P I(P,W)$ is the {\em channel capacity}. Since the 1960s,  there has been interest in determining finer asymptotic characterizations of the coding theorem. This is useful because such an analysis provides key insights into the amount of backoff from channel capacity for block codes of finite length $n$.  In particular, Strassen in 1962 \cite{Strassen} showed using normal approximations that the asymptotic expansion of $\log M_{\max}^*(W^n,\eps)$  satisfies
\begin{equation}
\log M_{\max}^*(W^n,\eps)=nC+\sqrt{nV_{\eps}} \Phi^{-1}(\eps)+\rho_n,  \label{eqn:sna}
\end{equation}
where $\rho_n=O(\log n)$,  $V_{\eps}$ is   the {\em  $\eps$-channel dispersion} \cite{PPV10, Pol10} and $\Phi(\cdot)$ is the Gaussian cumulative distribution function.\footnote{In fact, it was pointed out by Polyanskiy~\cite[Sec.\ 3.4.1]{Pol10} that Strassen's paper~\cite[Thm.~1.2]{Strassen} contains a gap in the case when the DMC is exotic and $\eps> \frac{1}{2}$. }   These quantities will be defined precisely in Section~\ref{sec:prelims}.
 In fact, this asymptotic expansion also holds for $M^*(W^n,\eps)$~\cite[Eqs.\ (284)-(286)]{PPV10} and implies that if an error probability of $\eps$ is tolerable, the backoff from channel capacity  $C$ at    finite blocklength $n$  is roughly $\sqrt{  {V_{\eps}}/{n}} \,  \Phi^{-1}(\eps)$. There have been  several  recent refinements to and extensions of Strassen's normal approximation in~\eqref{eqn:sna},  most prominently by Hayashi~\cite{Hayashi09} and \PPV~\cite{PPV10}. Strassen's normal approximation has also been shown to hold for many other classes of channels such as the additive white Gaussian noise (AWGN) channel \cite{Hayashi09, Pol10,PPV10}  and the additive Markovian channel~\cite{Hayashi09}.

Despite these impressive advances in the fundamental limits of channel coding, the third-order  term $\rho_n$   is not well understood.  Indeed, Hayashi in the conclusion of his paper~\cite{Hayashi09} mentions that 
\begin{quote}
{\em ``$\ldots$ the third-order coding rate is expected but appears difficult. The second order is the order $\sqrt{n}$, and it is not clear whether the third-order is a constant order or the order $\log n$'' }
\end{quote}
What we do know is that for  the binary symmetric channel (BSC),   $\rho_n=\frac{1}{2}\log n + O(1)$ \cite[Thm.\ 52]{PPV10} and for the  binary erasure channel (BEC),   $\rho_n=O(1)$ \cite[Thm.\ 53]{PPV10}. More generally,  there are classes of channels for which we have bounds on $\rho_n$ \cite[Sec.\ 3.4.5]{Pol10}. For lower bounds (achievability),  if we consider DMCs $W$ with positive reverse dispersion~\cite[Eq.~(3.296)]{Pol10}, then $\rho_n\ge\frac{1}{2}\log n  + O(1)$ \cite[Cor.~54]{Pol10}.  For   upper bounds (converse), if we restrict our attention to so-called {\em weakly input-symmetric}  DMCs \cite[Def.\ 9]{Pol10},   $\rho_n\le  \frac{1}{2}\log n + O(1)$ \cite[Thm.~55]{Pol10}.  For  {\em constant-composition codes}, it was shown~\cite{Mou12}  using strong large-deviation techniques~\cite{Cha93, Bahadur60}  that, under some regularity assumptions,  $\rho_n=\frac{1}{2}\log n + O(1)$.  Recall that a constant-composition code is one where all the codewords are of the same {\em empirical distribution} or {\em type}.   It is also claimed that the same holds for a more general class of DMCs in~\cite{Mou12a}. Our results generalize the converse bounds in~\cite{Mou12} and \cite{Mou12a}.

This paper strengthens the upper   bound  (converse) on the third-order  term  $\rho_n$.  For all DMCs whose $\eps$-dispersions are positive, we show that 
\begin{equation}
\log M^*(W^n,\eps)\le nC+\sqrt{nV_{\eps}} \Phi^{-1}(\eps)+\frac{1}{2}\log n +O(1), \label{eqn:main_res}
\end{equation}
If the $\eps$-dispersion vanishes, the corresponding bound is $\log M^*(W^n,\eps)\le nC+O(1)$, unless the DMC is exotic~\cite[Thm.\ 48]{PPV10} and $\eps \ge \frac{1}{2}$. If the DMC is exotic and $\eps=\frac{1}{2}$, we show that $\log M^*(W^n,\frac{1}{2})\le nC+\frac{1}{2}\log n+O(1)$. If  the DMC is exotic and $\eps>\frac{1}{2}$,  $\log M^*(W^n,\eps)\le nC+O\big(n^{\frac{1}{3}}\big)$,  a result by \PPV~\cite[Thm.\ 48]{PPV10}. 
Hence,  for the rather general class of DMCs with positive $\eps$-dispersion, the third-order  term  is $\rho_n\le\frac{1}{2}\log n + O(1)$. 
We may thus dispense with the   assumption that $W$ is weakly input-symmetric \cite[Def.\ 9]{Pol10}. 

The typical  way~\cite{Strassen, PPV10, Pol10} to  upper bound    $M^*(W^n,\eps)$ is to first do the same for the maximum size of a constant-composition  code   under the maximum error probability formulation.  Such a bound can be proved using either the meta-converse \cite[Thm.\ 31]{PPV10} or tight bounds on the type-II error probability  in a simple binary hypothesis test \cite[Thm.\ 1.1]{Strassen}. By the type-counting lemma \cite[Lem.\ 2.2]{Csi97}, every length-$n$ block code can be partitioned into no more than $(n+1)^{|\cX|-1 }$ constant-composition subcodes. This leads to the rather conservative bound \cite[Eq.\ (4.29)]{Strassen} \cite[Eq.\ (279)]{PPV10}
\begin{equation}
\log M_{\max}^*(W^n,\eps)\le nC+\sqrt{nV_{\eps}} \Phi^{-1}(\eps)+\Big(|\cX|-\frac{1}{2} \Big)\log n +O(1).   \label{eqn:types_res}
\end{equation}
Subsequently,  by expurgating bad codewords (see \cite[Eqs.\ (284)-(286)]{PPV10}), we can conclude that the same upper bound holds for  $M^*(W^n,\eps)$.
We adopt a different approach for the proof of our main result in \eqref{eqn:main_res} and work with $M^*(W^n,\eps)$ directly.  In a nutshell, we consider a new ``symbol-wise'' relaxation of the meta-converse that allows us to work directly  with general (non-constant-composition) codes and the average probability of error. The one-shot converse is stated in terms of the {\em relative entropy information spectrum} \cite[Ch.~4]{Han10} but allows us to choose an auxiliary output distribution as in the meta-converse. We then
 carefully weigh the contributions of each input type for  a general code by constructing an appropriate $\epsilon$-net for the output probability simplex.  The last step, which replaces the use of the  type-counting lemma, is one of our main contributions and allows us to bound the effect of different  input types with the $O(1)$ term in~\eqref{eqn:main_res}. 

Note that unlike in \eqref{eqn:types_res}, the third-order term in our upper bound in \eqref{eqn:main_res} is independent of $|\cX|$. This is intuitively plausible due to the following observation. Let $n$ be a large even integer and consider using transmitting information across $n$ uses of a DMC $W:\cX\to\cY$. Clearly, the same amount of information can be transmitted through $\frac{n}{2}$ uses of the product channel $W^{2} :\cX^{\cp 2}\to\cY^{\cp 2}$, where $W^{2}(y,y'|x,x') :=W(y|x) W(y'|x')$. The capacity and the dispersion of $W^{2}$ are respectively twice the capacity and the dispersion of $W$  so the normal approximation terms for $n$ uses of $W$ and $\frac{n}{2}$ uses of $W^{2}$ are identical.  If the  coefficient of the third-order logarithmic term   {\em were} dependent on the size of the input alphabet, say via some function $g(|\cX|)$, then in the first case, $\rho_n =  g(|\cX|)\log n+O(1)$ while in the second case, $\rho_n = g(|\cX|^2)\log (\frac{n}{2})+O(1)= g(|\cX|^2)\log  n+O(1)$. Thus,  at least on an intuitive level, we expect that $g(|\cX|)$ is independent of $|\cX|$.

\section{Notation and Preliminaries}

\subsection{Discrete Memoryless Channels} \label{sec:prelims}

As mentioned in the Introduction, we consider \emph{discrete memoryless channels} (DMCs), which are characterized by two finite sets, the input alphabet $\cX$ and the output alphabet  $\cY$, and a stochastic matrix $W$, where $W(y|x)$ denotes the probability that the output $y \in \cY$ occurs given input $x \in \cX$. 
The set of probability distributions on $\cX$ is denoted $\cP(\cX)$.
For any probability distribution $P \in \cP(\cX)$, we denote by $P \cp W : (x,y) \mapsto P(x) W(y|x)$ the joint distribution of inputs and outputs of the channel, and by $PW : y \mapsto \sum_x P(x) W(y|x)$ its marginal on $\cY$. Finally, $W(\cdot|x)$ denotes the distribution on $\cY$ if the input is fixed to $x$.

Given two probability distributions $P, Q \in \cP(\cX)$, we call the random variable $\log \frac{P(X)}{Q(X)}$ where $X$ has distribution $P$ the \emph{log-likelihood ratio} of $P$ and $Q$. Its mean is the {\em relative entropy} 
\begin{align*}
  D(P \| Q) := \Exp_P \bigg[\log \frac{P}{Q}\bigg] = \sum_{x \in \cX} P(x) \log \frac{P(x)}{Q(x)} 
\end{align*}
and $D(W\|Q|P):=\sum_x P(x)D(W(\cdot|x)\|Q)$ is the {\em conditional  information divergence}. The {\em mutual information} is $I(P, W) := D( W \| PW|P)$. Moreover, 
\begin{align}
  C(W) := \max_{P \in \cP(\cX)} I(P, W) \qquad \textrm{and} \qquad 
  \Pi(W) := \{P \in \cP(\cX) \, |\, I(P, W)=C(W) \} \nonumber
\end{align}
are the \emph{capacity} and the set of \emph{capacity-achieving input distributions} (CAIDs), respectively.\footnote{We often drop the dependence on $W$ if it is clear from context.} The set of CAIDs is convex and compact in $\cP(\cX)$. 
The unique~\cite[Cor.\ 2 to Thm.\ 4.5.2]{gallagerIT}  \emph{capacity-achieving output distribution} (CAOD) is denoted as $Q^*$ and $Q^* = P W$ for all
$P \in \Pi$. Furthermore, it satisfies $Q^*(y) > 0$ for all $y \in \cY$ \cite[Cor.\ 1 to Thm.\ 4.5.2]{gallagerIT}, where we assume that all outputs are accessible.

The variance of the log-likelihood ratio of $P$ and $Q$ is the \emph{divergence variance}
\begin{align*}
  V(P \| Q) := \Exp_P \bigg[ \Big( \log \frac{P}{Q} - D(P \| Q) \Big)^2 \bigg] .
\end{align*}
We also define the {\em conditional divergence variance} $V( W \| Q | P) := \sum_x P(x) V( W(\cdot|x) \| Q)$ and   the {\em conditional information variance}
$V(P, W) := V(W \| PW | P)$. Note that $V(P, W) = V(P \cp W \| P \cp PW)$ for all $P \in \Pi$ \cite[Lem.\ 62]{PPV10}.  The $\eps$-\emph{channel dispersion}\footnote{Notice that for $\eps=\frac{1}{2}$, we set $V_\eps=V_{\max}$. This is somewhat unconventional; cf.~\cite[Thm.\ 48]{PPV10}. However, doing so ensures that  Theorem~\ref{th:main} can be  stated compactly. Nonetheless, from the viewpoint of the normal approximation,  it is immaterial how we choose $V_{\frac{1}{2}}$  since $\Phi^{-1}   (\frac{1}{2})=0$ (cf. \cite[after Eq.\ (280)]{PPV10}). }~\cite[Def.\ 2]{PPV10} is an operational quantity that was shown~\cite[Eq.\ (223)]{PPV10} to be equal to
\begin{align}
  V_{\eps}(W) := \begin{cases} V_{\min}  & \textrm{if } \eps < \frac{1}{2} \\ V_{\max} & \textrm{if } \eps \geq \frac{1}{2} \end{cases}, \quad \textrm{where} \quad 
  V_{\min} := \min_{P \in \Pi} V(P, W) \quad  \textrm{and} \quad V_{\max} := \max_{P \in \Pi} V(P, W) \,. \nonumber
\end{align}

Furthermore, a channel is called \emph{exotic}~\cite[before Thm.~48]{PPV10} if $V_{\max} = 0$ and there exists a symbol~$x_0  \in\cX$ such that $D( W(\cdot|x_0 ) \| Q^*) = C$ and $V( W(\cdot|x_0) \| Q^*) > 0$.\footnote{Note that this symbol must
satisfy $P(x_0) = 0$ for any $P \in \Pi$, as otherwise $V_{\max}$ would not vanish.}

For later reference, we also define the {\em  third absolute  moment of the log-likelihood ratio},
\begin{align*}
  T(P\|Q) := \Exp_P \bigg[ \Big| \log \frac{P}{Q} - D(P \| Q) \Big|^3 \bigg]
\end{align*}
and
$T(W\|Q|P) := \sum_x P(x) T( W(\cdot|x) | Q)$.
 
We employ the cumulative distribution function of the standard normal distribution 
$$\Phi(a) := \int_{-\infty}^a \frac{1}{\sqrt{2 \pi}} \exp \Big(-\frac12 x^2\Big) \,\mathrm{d}x$$ 
and define its inverse as $\Phi^{-1}(\eps) := \sup \{ a \in \mathbb{R} \,|\, \Phi(a) \leq \eps \}$, which evaluates to the usual inverse for $0 < \eps < 1$ and continuously extends to take values $\pm \infty$ outside that range.

For a sequence   $\vec{x} = (x_1, x_2, \ldots, x_n) \in \cX^{\times n}$, we denote by $P_{\vec{x}} \in \cP(\cX)$ the probability distribution given by the relative frequencies of 
$\vec{x}$, i.e.\ $P_{\vec{x}}(x) = \frac{1}{n} \sum_{i=1}^n 1 \{x_i=x\}$.  This probability distribution $P_{\vec{x}}$ is also known as the {\em empirical distribution} or the {\em type} \cite[Def.~2.1]{Csi97} of $\vec{x}$.
The set of all such distributions is denoted  as $\cP_n(\cX) = \bigcup_{\vec{x}} \big\{ P_{\vec{x}} \big\}$ and satisfies $| \cP_n(\cX) | \leq (n+1)^{|\cX|-1}$.

\subsection{Codes and $\eps$-Error Capacity}

A {\em code} $\cC$ for   a channel is defined by the triple $\{\cM, e, d\}$,
where $\cM$ is a set of messages, $e: \cM \to \cX$ an encoding function and $d: \cY \to \cM$ a decoding function. We write $\abs{\cC} = \abs{\cM}$ for the cardinality of the message set. We define the \emph{average error probability} of a code $\cC$ for the channel $W$  as 
$$
\perr(\cC,W):= P[M \neq M'] = 1-\frac{1}{ |\cM|}\sum_{m\in\cM}W(d^{-1}(m)|e(m))
$$ 
where the distribution over messages $P_M$ is assumed to be uniform on $\cM$, 
$$M \xrightarrow{\ e\ } X \xrightarrow{\ W\ } Y \xrightarrow{\ d\ } M'$$
forms a Markov chain, and $M'$ thus denotes output of the decoder. The \emph{one-shot $\eps$-error capacity} of the channel $W$ is then defined as
\begin{align*}
  M^*(W,\eps) := \max \big\{ m \in \mathbb{N} \,\big|\, \exists\, \cC :\ \abs{\cC} = m \ \land\ \perr(\cC, W) \leq \eps \big\} .
\end{align*}

We are also interested in the $\eps$-error capacity for $n \geq 1$ uses of a memoryless channel. For this purpose, we consider the channel $W^n:\cX^n\to\cY^n$, defined by the stochastic matrix 
$W^n(\vec{y}|\vec{x}) = \prod_{i=1}^n W(y_i|x_i)$, where $\vec{x} = (x_1, x_2, \ldots, x_n)$ and $\vec{y} =( y_1, y_2, \ldots, y_n)$ are strings of length $n$ of symbols $x_i \in \cX$ and $y_i \in \cY$, respectively. Then, the \emph{blocklength $n$, $\eps$-average error capacity} of the channel  $W$ is denoted as $M^*(W^n,\eps)$.

\begin{figure}
\centering
\begin{picture}(384,120)
\thicklines
\put(0, 75){\line(0,1){30}}
\put(40, 75){\line(0,1){30}}
\put(0, 75){\line(1,0){40}}
\put(0, 105){\line(1,0){40}}
\thinlines

\put(40, 90){\vector(4,3){40}}
\put(40, 90){\vector(4,-3){40}}
\put(50, 110){\footnotesize\mbox{Yes}}
\put(50, 67){\footnotesize\mbox{No}}
\put(7, 88){\footnotesize \mbox{$V_\eps>0$}}

\put(83,117){\footnotesize\mbox{$\le\! nC \!+\!\sqrt{nV_\eps}\Phi^{-1}(\eps)\! +\! \frac{1}{2}\log n \!+\! O(1)$ [Props.\ \ref{pr:reg} and  \ref{prop:epshalf}(i)] }}

\thicklines
\put(80, 45){\line(0,1){30}}
\put(120, 45){\line(0,1){30}}
\put(80, 45){\line(1,0){40}}
\put(80, 75){\line(1,0){40}}
\thinlines

\put(120, 60){\vector(4,3){40}}
\put(120, 60){\vector(4,-3){40}}
\put(130, 80){\footnotesize\mbox{Yes}}
\put(130, 37){\footnotesize\mbox{No}}
\put(81, 65){\footnotesize \mbox{not $\!\!\!$ exotic}}
\put(84, 55){\footnotesize \mbox{or $\eps\!<\!\frac{1}{2}$}}

\put(163,87){\footnotesize\mbox{$\le\! nC\! +\! O(1)$ [Prop.\ \ref{pr:special}] }}

\thicklines
\put(160, 15){\line(0,1){30}}
\put(200, 15){\line(0,1){30}}
\put(160, 15){\line(1,0){40}}
\put(160, 45){\line(1,0){40}}
\thinlines

\put(200, 30){\vector(4,3){40}}
\put(200, 30){\vector(4,-3){40}}
\put(210, 50){\footnotesize\mbox{Yes}}
\put(210, 7){\footnotesize\mbox{No}}
\put(167, 34){\footnotesize \mbox{exotic}}
\put(162, 24){\footnotesize \mbox{and $\!\!$  $\eps\!=\!\frac{1}{2}$}}

\put(243, 57){\footnotesize\mbox{$\le\! nC\! +\! \frac{1}{2}\log n\! +\! O(1)$ [Prop.\ \ref{prop:epshalf}(ii)] }}
\put(243, 0){\footnotesize\mbox{$\le\! nC\! +\!   O\big(n^{\frac{1}{3}}\big)$  \cite[Thm.\ 48]{PPV10}  }}
\end{picture}
\caption{Illustration of the various cases of Theorem~\ref{th:main} and  the proof structure in Section~\ref{sec:asymp}}
\label{fig:flowchart}
\end{figure}
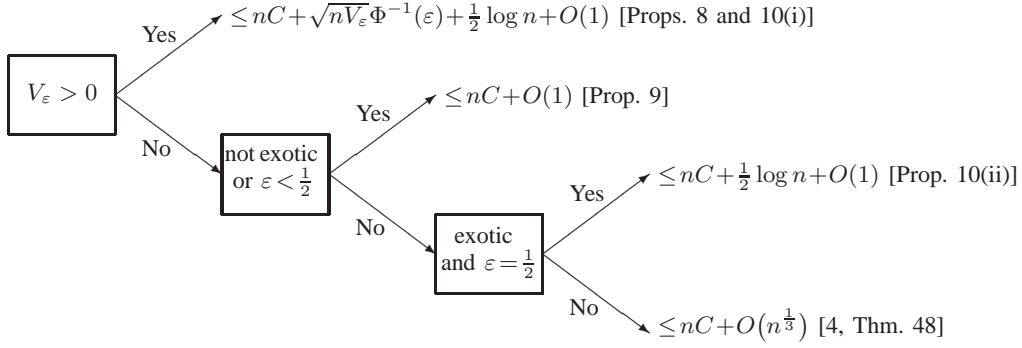

\section{Main Result and Proof}

Let us reiterate our main result.  The various cases are illustrated diagrammatically in Fig.~\ref{fig:flowchart}.
\begin{theorem}
\label{th:main}
For every DMC $W$ and $\eps$ with $V_{\eps} > 0$, the  blocklength $n$, $\eps$-error capacity satisfies
\begin{equation}
\log M^*(W^n,\eps)\le nC+\sqrt{nV_{\eps}} \Phi^{-1}(\eps)+\frac{1}{2}\log n + O(1).  \nonumber
\end{equation}
If $V_{\eps} = 0$, we have $\log M^*(W^n,\eps)\le nC + O(1)$, unless the channel is exotic and   $\eps \ge  \frac{1}{2}$.
\end{theorem}

\begin{remark} 
The $\eps=\frac{1}{2}$ case needs to be treated with care.  
For all DMCs $W$ with $V_{\min } = 0$ and $\eps=\frac{1}{2}$ (this includes exotic DMCs), we show that  $\log M^*(W^n,\eps)\le nC + \frac{1}{2}\log n+ O(1)$. See Proposition~\ref{prop:epshalf}. If $V_{\max}>0$, this statement  concurs with the positive  $\eps$-dispersion case of Theorem~\ref{th:main}.
\end{remark}

\begin{remark}
From the preceding statements, we see that for DMCs  with $V_{\min} = 0$ and $V_{\max} > 0$, the third-order term ``jumps'' from $0$ to $\frac{1}{2} \log n$ when $\eps\uparrow\frac{1}{2}$.  This is possible because we do not investigate the dependence of the constant term on $\eps$.\footnote{Indeed, in our proof for the case $V_{\min}=0$, $V_{\max} > 0$ and $\eps= \big(\frac{1}{2}\big)^-$ in~Proposition~\ref{pr:special}, we notice that the constant term diverges as $\eps\uparrow \frac{1}{2}$.}
\end{remark}

In light of the existing results on $\rho_n$ (in the Introduction and \cite[Sec.\ 3.4.5]{Pol10}), the third-order term is the best possible unless we impose further assumptions on $W$. More precisely, it was shown in \cite[Cor.~54]{Pol10} that if there exists a $P \in \Pi(W)$ achieving $V_\eps(W)$ such that the {\em reverse conditional information variance} is positive, i.e.\ $V^{\mathrm{r}}(P,W) := V\big(PW, \frac{P \times W}{PW}\big)>0$, then 
\begin{align*}
\log M^*(W^n,\eps)\ge nC+\sqrt{nV_{\eps}} \Phi^{-1}(\eps)+\frac{1}{2}\log n + O(1).
\end{align*}
This matches the upper bound of Theorem~\ref{th:main}.

The proof consists of five parts, each detailed in one of the following subsections. In the first subsection, we introduce two entropic quantities, the hypothesis testing divergence~\cite{Wang09, WangRenner, Tom12, Dupuis12} and a quantity related to the information (or divergence) spectrum~\cite[Ch.\ 4]{Han10}. We state and prove some useful and well-known properties that we need later. In the second subsection, we derive a converse bound, valid for general DMCs, that involves a minimization over output distributions and maximization over input symbols. In the third subsection, we choose an appropriate output distribution for use in the general converse bound. In the fourth subsection, we state and prove some continuity properties of information measures around the CAIDs and the unique CAOD. Finally, the fifth subsection contains the proof of our main result.

\subsection{Hypothesis Testing and the Information Spectrum}

We use the following divergence~\cite{Wang09, Tom12,WangRenner, Dupuis12}, which is closely related to binary hypothesis testing. Let $\eps \in (0,1)$ and let $P, Q \in \cP(\cZ)$, where $\cZ$ is finite. We consider binary (probabilistic) hypothesis tests $\xi : \cZ \to [0,1]$ and define the \emph{$\eps$-hypothesis testing divergence}
\begin{align*}
 D_h^{\eps}(P \| Q) := \sup \Big\{ R \in \mathbb{R} \,\Big|\, \exists\ \xi :\ 
   \Exp_{Q} \big[\xi(Z)\big] \leq (1-\eps) \exp(-R)\ \land\ \Exp_{P} \big[\xi(Z)\big] \geq 1-\eps  \Big\}.
\end{align*}
Note that $D_h^{\eps}(P \| Q)=-\log \frac{\beta_{1-\eps}(P,Q)}{1-\eps}$  where $\beta_{1-\eps}(P,Q)$ is the smallest type-II error of a hypothesis test between $P$ and $Q$ with type-I error smaller than $\eps$ and is defined formally in~\cite[Eq.\ (100)]{PPV10}. It is easy to see that $D_h^{\eps}(P\|Q)\geq 0$, where the lower bound is achieved if and only if $P = Q$ and $D_h^{\eps}(P \| Q)$ diverges if $P$ and $Q$ are orthogonal. It satisfies a data-processing inequality~\cite{Wang09}
\begin{align*}
   D_h^{\eps}(P \| Q) \geq D_h^{\eps}(PW \| QW) 
   \qquad \textrm{for all channels $W$ from $\cZ$ to $\cZ'$} .
\end{align*}
When evaluated for independent and identical distributions (i.i.d.), its asymptotic expansion in the first order is determined by the Chernoff-Stein Lemma~\cite[Cor.\ 1.2]{Csi97}, yielding
$D_h^{\eps}(P^{\times n}\|Q^{\times n}) = n D(P\|Q) + o(n)$ for any $\eps \in (0,1)$. This asymptotic expansion was subsequently tightened by Juschkewitsch~\cite{Jus53} among others. Finally Strassen~\cite[Thm.\ 3.1]{Strassen} found an expansion including the third-order term as
\begin{align*}
  D_h^{\eps}(P^{\times n} \| Q^{\times n}) = n D(P\|Q) + \sqrt{n V(P\|Q)} \Phi^{-1}(\eps) + \frac{1}{2} \log n + O(1) .
\end{align*}
%

The following quantity, which characterizes the distribution of the log-likelihood ratio and is
known as the {\em relative entropy information spectrum} or the {\em divergence spectrum}  \cite[Ch.~4]{Han10}, is sometimes easier to manipulate and evaluate.
  \begin{align*}
    D_s^{\eps}(P \| Q) := \sup \bigg\{ R \in \mathbb{R} \,\bigg|\, 
    P \Big[ \log \frac{P}{Q} \leq R \Big] \leq \eps \bigg\} .
  \end{align*}
It is intimately related to the $\eps$-hypothesis testing divergence.
\begin{lemma} 
  \label{lm:h-s}
For any $\delta \in (0, 1-\eps)$, we have
\begin{align}
  \label{eq:Ds-Dh}
   D_h^{\eps}(P \| Q) \leq D_s^{\eps+\delta}(P \| Q) + \log \frac{1-\eps}{\delta} .
\end{align}
\end{lemma}
\noindent  
This relation follows from standard arguments relating binary hypothesis testing and the log-likelihood test to the relative entropy information spectrum. 
See, for example~\cite[Eq.~(102) and Eqs.~(158)-(159)]{PPV10} where this is used to relax the meta-converse to (a generalization of) the Verd\'u-Han information spectrum converse~\cite[Lem.~3.2.2]{Han10} or~\cite[Lem.\ 12]{Tom12}, where an analogue of the above lemma is shown for the strictly more general non-commutative case.

We can give an upper bound on $D_s^{\eps}(P\|Q)$ if $Q$ is a convex combination of distributions.
\begin{lemma}
\label{lm:convex}
  Let $P \in \cP(\cZ)$ and $Q = \sum_{i \in \cI} q(i) Q^i$ with $Q^i \in \cP(\cZ)$ and $q \in \cP(\cI)$ and $\cI$ is some countable index set. Then,
\begin{align*} 
  D_s^{\eps}(P \| Q) \leq \inf \big\{ D_s^{\eps}(P \| Q^i) - \log q(i)  \big\}_{i \in \cI}
\end{align*}
\end{lemma}
\begin{proof}
  Note that for all $z \in \cZ$ with $P(z) > 0$, for all $i \in \cI$, we have
  \begin{align*}
    \log \frac{P(z)}{Q(z)} = \log \frac{P(z)}{\sum_j q(j) Q^j(z)} \leq \log \frac{P(z)}{q(i) Q^i(z)}
    = \log \frac{P(z)}{Q^i(z)} - \log q(i) .
  \end{align*}
  Hence,
  \begin{align*}
    P \bigg[ \log \frac{P}{Q} \leq R \bigg] \geq  P \bigg[ \log \frac{P}{Q^i} \leq R + \log q(i) \bigg] 
  \end{align*}
  and thus we find $D_s^{\eps}(P\|Q) \leq D_s^{\eps}(P\|Q^i) - \log q(i)$ for any $i \in \cI$ as desired.
\end{proof}

The following standard result will be particularly useful, as it allows us to bound
the log-likelihood ratio of the input-output behavior of two channels in terms of
the log-likelihood ratio evaluated for a single input symbol.
\begin{lemma}
  \label{lm:channels}
  Let $P \in \cP(\cX)$ and let $V,\, W$ be channels from $\cX$ to $\cY$. Then,
  \begin{align*}
    D_s^{\eps}(P \cp W\|P \cp V) \leq \sup_{x:\, P(x) > 0}\ D_s^{\eps}( W(\cdot|x) \| V(\cdot|x) ) .
  \end{align*}
\end{lemma}

\begin{proof}
We first note that the log-likelihood ratio takes on the form
\begin{align*}
  \log \frac{P \cp W}{P \cp V} :\ (x,y)\ \mapsto\ \log \frac{P(x)W(y|x)}{P(x)V(y|x)} 
  = \log \frac{W(y|x)}{V(y|x)} ,
\end{align*}
for every $(x,y) \in\cX\times\cY$ satisfying $P(x)>0$. 
Now, we may write
\begin{align*}
  R^* &= D_s^{\eps}(P \cp W\|P \cp V) = \sup \bigg\{ R \in \mathbb{R} \,\bigg|\, 
    P \Big[ \log \frac{P \cp W}{P \cp V} \leq R \Big] \leq \eps \bigg\} \\
    &= \sup \bigg\{ R \in \mathbb{R} \,\bigg|\, 
    \sum_{x : P(x)>0} P(x)\, W \Big[ \Big\{ y\,\big|\, \log \frac{W(y|x)}{V(y|x)} \leq R  \Big\}\,\Big|\,   x \Big] \leq \eps \bigg\} .
\end{align*}
Inspecting this expression, for any $\varphi > 0$, we find at least one $x^* \in \mathcal{X}$ such that 
\begin{align*}
P(x^*) > 0 \quad \textrm{and} \quad 
 W \Big[ \Big\{ y\,\big|\, \log \frac{W(y|x)}{V(y|x)} \leq R  \Big\}\,\Big|\,   x \Big] \leq    \eps\, .
\end{align*}
Hence, $D_s^{\eps}(W(\cdot|x^*)\|V(\cdot|x^*)) \geq R^* - \varphi$, which implies the lemma as $\varphi$ is arbitrary.
\end{proof}

The distribution of the log-likelihood ratio has the following asymptotic expansions for not necessarily identical product distributions. The bounds follow from   simple applications of the Berry-Essen theorem~\cite[Sec.\ XVI.5]{feller} and Chebyshev's inequality.
\begin{lemma}
  \label{lm:Ds-converse}
  Let $P_i, Q \in \cP(\cZ)$ be such that $Q$ dominates $P_i$ for all $i$ in some finite set $\mathcal{I}$. We consider a sequence of distributions $P_{i_k}$ indexed by  $(i_1, i_2, \ldots, i_n )$ where $i_k\in\cI$ for each $1\le k\le n$. Define
  \begin{align*}
   D_n := \frac{1}{n} \sum_{k=1}^n D(P_{i_k} \| Q), \,\,\,
   V_n := \frac{1}{n} \sum_{k=1}^n V(P_{i_k} \| Q), \,\,\, \textrm{and} \,\,\,\,
   T_n  := \frac{1}{n}\sum_{k=1}^n T(P_{i_k} \| Q)\, .
  \end{align*}
  If $V_n > 0$, then we have the Berry-Esseen-type bound
  \begin{align*}
    D_s^{\eps} \big( P_{i_1} \cp \ldots  P_{i_n} \big\| Q^{\times n}\big) 
    \leq n D_n + \sqrt{n V_n} \Phi^{-1}\bigg( \eps + \frac{6 \, T_n}{\sqrt{n V_n^3}} \bigg)  .
  \end{align*}
  In any case, we have the Chebyshev-type bound
  \begin{align}
    D_s^{\eps} \big( P_{i_1} \cp \ldots  P_{i_n} \big\| Q^{\times n}\big) \leq n D_n + \sqrt{\frac{n V_n}{1-\eps}} .  \label{eq:cheby}
  \end{align}
\end{lemma}

\begin{proof}
  We consider  the cumulative distribution of the random variable $S_n:=\sum_{k} \log P_{i_k}(X_{i_k}) - \log Q(X_{i_k})$ where each $X_{i_k}$ has distribution $P_{i_k}$. The random variable $S_n$ has mean $n D_n$ and variance $n V_n$.
  The general case, Eq.~\eqref{eq:cheby}, is shown using Chebyshev's inequality, which yields
  \begin{align*}
    \eps \geq P \bigg[ \sum_{k} \log \frac{P_{i_k}}{Q} \leq R \bigg] \geq 1 - \frac{n V_n}{(R - n D_n )^2} \qquad \textrm{for } R > n D_n
  \end{align*}
  Hence, restricting to $R > n D_n$ and relaxing the bound on $R$ in the supremum, we find
  \begin{align*}
     D_s^{\eps} \big( P_{i_1} \cp \ldots  P_{i_n} \big\| Q^{\times n} \big)
      \leq \sup \Big\{ R > n D_n \,\Big|\, 1 - \frac{n V_n}{(R - n D_n)^2} \leq \eps \Big\} = n D_n + \sqrt{\frac{n V_n}{1 - \eps}} .
  \end{align*}
  Furthermore, if $V_n > 0$, the Berry-Esseen theorem~\cite[Sec.\ XVI.5]{feller} states that
  \begin{align*}
    \Bigg| P \bigg[ \sum_{k} \log \frac{P_{i_k}}{Q} \leq R \bigg] - \Phi\bigg( \frac{R - n D_n}{\sqrt{n V_n}} \bigg) \Bigg| \leq \frac{6\, T_n}{\sqrt{n V_n^3}} .
  \end{align*}
  Hence, we obtain
  \begin{align*}
    D_s^{\eps} \big( P_{i_1} \cp \ldots  P_{i_n} \big\| Q^{\times n} \big)
    &\leq n D_n + \sqrt{n V_n} \Phi^{-1}\bigg(\eps + \frac{6\, T_n}{\sqrt{n V_n^3}}\bigg),
  \end{align*}
  which concludes the proof.
\end{proof}

\subsection{Converse Bounds on General Channels}

Here, we give a new converse bound on the size of arbitrary codes for general channels, for the average probability of error formulation. 

\begin{proposition}
  \label{th:one-shot}
  Let $\eps \in (0, 1)$ and let $W$ be any channel from $\cX$ to $\cY$. 
  Then, for any $\delta \in (0, 1-\eps)$, we have
  \begin{align*}
    \log M^*(W,\eps) &\leq \inf_{Q \in \cP(\cY)}\ \sup_{x \in \cX}\ D_s^{\eps+\delta}\big( W(\cdot|x) \big\| Q \big) + 
    \log \frac{1}{\delta}  .
  \end{align*}
\end{proposition}

The first part of the proof is analogous to the meta-converse in~\cite[Thm.~27]{PPV10} (see also~\cite{Wang09} and~\cite{WangRenner}, which inspired our conceptually simpler proof technique). Our bound is a new ``symbol-wise''  relaxation of the meta-converse  
which yields a result in the spirit of~\cite[Thms.~28 and~31]{PPV10}. The maximization over symbols allows us to apply our converse bound on non-constant-composition codes directly.

\begin{proof}
  For any code $\cC = \{\mathcal{M}, e, d\}$ with $\perr(\cC) \leq \eps$ 
  and any $Q \in \cP(\cY)$, the following holds. 
  
  Starting from a uniform distribution over $\cM$, the Markov chain 
  $M \xrightarrow{\ e\ } X \xrightarrow{\ W\ } Y \xrightarrow{\ d\ } M'$ induces a
  joint probability distribution $P_{MXYM'}$. 
  Due to the data-processing inequality for $D_h^\eps$, we immediately
  find 
  $D_h^{\eps}(P \cp W\|P \cp Q) = D_h^{\eps}(P_{XY} \| P_X \cp Q_Y) \geq D_h^{\eps}(P_{MM'} \| P_M \cp Q_{M'})$,  
  where $P_X=P$ and $Q_{M'}$ is the distribution induced by $d$ applied to $Q_Y = Q$.\footnote{Note that due to the Markov property, the encoding can be inverted probabilistically, without effecting the correlation between $M$ and $M'$.}
  Moreover, using the test $\xi(m,m') = \delta_{m,m'}$, we readily see that
  \begin{align*}
  \Exp_{P \times W} \big[\xi(M,M')\big] = P [ M = M']\geq 1 - \eps 
  \quad \textrm{and} \quad
  \Exp_{P \times Q} \big[ \xi(M,M') \big] = \frac{1}{\abs{\cC}} .
  \end{align*}
  Hence, $D_h^{\eps}(P_{MM'} \| P_M \cp Q_{M'}) \geq \log \abs{ \cC} + \log (1\!-\!\eps)$ by definition of the $\eps$-hypothesis testing divergence.
  Finally, applying Lemmas~\ref{lm:h-s} and~\ref{lm:channels}, we find
  \begin{align*}
    \sup_{x \in \cX} D_s^{\eps+\delta}\big(W(\cdot|x) \big\| Q\big) &\geq 
    D_s^{\eps+\delta}\big(P \cp W \big\| P \cp Q\big) \\
    &\geq D_h^{\eps}\big(P \cp W \big\| P \cp Q\big) - \log \frac{1\!-\!\eps}{\delta}
    \geq \log \abs{\cC} -
    \log \frac{1}{\delta} .
  \end{align*}
 This yields the converse bound upon minimizing over $Q \in \cP(\cY)$. 
\end{proof}

\subsection{A Suitable Choice of Output Distribution $Q$}
\label{sc:output}

For $n$-fold repetitions of a DMC, the bound in Proposition~\ref{th:one-shot} evaluates to
\begin{align}
\log  M^*(W^n, \eps) \leq \min_{Q^{(n)} \in \cP(\cY^{\times n})}\ \max_{\vec{x} \in \cX^{\times n}} D_s^{\eps+\delta}\big( W^n(\cdot|\vec{x}) \big\|Q^{(n)}\big) + 
    \log \frac{1}{\delta} ,
\nonumber
\end{align}
and it is thus important to find a suitable choice of $Q^{(n)} \in \cP(\cY^{\times n}) $ to further upper bound the above. Symmetry considerations (see, e.g., \cite[Sec.~V]{Polyanskiy13}) allow us to restrict the search to distributions that are invariant under permutations of the $n$ channel uses.
\begin{figure}
\centering
\begin{picture}(150,150)
\put(0,5){\vector(1,0){150}}
\put(5,0){\vector(0,1){150}}

\put(150,10){\mbox{$Q(0)$}}
\put(10,150){\mbox{$Q(1)$}}
\put(67.5,67.5){\circle*{4}}
\put(71,67.5){\mbox{$Q^*$}}
\put(87.5,47.5){\circle*{4}}
\put(107.5,27.5){\circle*{4}}
\put(127.5,7.5){\circle*{4}}

\put(47.5,87.5){\circle*{4}}
\put(27.5,107.5){\circle*{4}}
\put(7.5,127.5){\circle*{4}}

\put(67.5,67.5){\line(0,-1){20}}
\put(67.5,47.5){\line(1,0){20}}

\put(69,38){\mbox{$\frac{1}{\sqrt{2n}}$}}
\put(45,56){\mbox{$\frac{1}{\sqrt{2n}}$}}

\put(51,87.5){\mbox{$Q_{[-1,1]}$}}
\put(87.5,51){\mbox{$Q_{[1,-1]}$}}

\put(110,27.5){\mbox{$Q_{[2,-2]}$}}

\put(33,106){\mbox{$Q_{[-2,2]}$}}

\put(45,125){\mbox{$\cP(\cY)$}}
\put(44,125){\vector(-3,-1){25}}

\put(125,-5){\mbox{$(0,1)$}}
\put(-19,125){\mbox{$(1,0)$}}
\thicklines
\put(5,130){\line(1,-1){125}}
\end{picture}
\caption{Illustration of the choice of $Q_{\vec{k}}$ for $\cY=\{0,1\}$. Note that $\zeta  =2$ for $|\cY|=2$.}
\label{eqn:quantize}
\end{figure}
Let $\zeta :=\abs{\cY} (\abs{\cY} - 1)$ and let $\gamma>0$ be a constant which is to be chosen later. Consider the following convex combination of product distributions:
\begin{align}
  Q^{(n)}(\vec{y}) := \frac12 \sum_{\vec{k} \in \mathcal{K}}
  \frac{\exp \big(- \gamma \twonorm{\vec{k}}^2 \big)}{F}
  \, \prod_{i=1}^n Q_{\vec{k}}(y_i) 
  + \frac12 \sum_{P_{\vec{x}} \in \cP_n(\cX)} \frac{1}{\abs{\cP_n(\cX)}} \prod_{i=1}^n P_{\vec{x}}W(y_i),  \label{eqn:Qn}
\end{align}
where $F$ is a normalization constant that ensures $\sum_{\vec{y}}Q^{(n)}(\vec{y}) =1$ and 
\begin{align*}
  Q_{\vec{k}}(y) := Q^*(y) + \frac{k_{y}}{\sqrt{n \zeta }}, \qquad
  \mathcal{K} := \Big\{ \vec{k} \in \mathbb{Z}^{\abs{\cY}} \,\Big|\, \sum_y k_y = 0 \land k_y \geq - Q^*(y) \sqrt{n \zeta } \Big\}.
\end{align*}
The convex combination of $(P_{\vec{x}}W)^{\times n}$ and the optimal  output distribution $(Q^*)^{\times n}$ (corresponding to $\vec{k}=\vec{0}$) in $Q^{(n)}$ is  inspired partly  by Hayashi~\cite[Thm.\ 2]{Hayashi09}.  What we have done in our choice of $Q_{\vec{k}}$ is   to uniformly quantize the simplex $\cP(\cY)$ along axis-parallel directions.   The constraint  that each $\vec{k}$ belongs to $\mathcal{K}$ ensures that each $Q_{\vec{k}}$ is a valid probability mass function. See Fig.~\ref{eqn:quantize}.  We find that
\begin{align*}
  F \leq \sum_{\vec{k} \in \mathbb{Z}^{\abs{\cY}}} \exp \big(-\gamma \twonorm{\vec{k}}^2 \big)
  = \Bigg( \sum_{k= -\infty}^{\infty} \exp\big(- \gamma k^2\big) \Bigg)^{\abs{\cY}}
  \leq \Bigg( 1 + \sqrt{\frac{\pi}{\gamma}} \Bigg)^{\abs{\cY}}
\end{align*}
is a finite constant. Furthermore, by construction, the representation points  $\{ Q_{\vec{k}} \}_{\vec{k}}$ form an 
{\em $\epsilon$-net} with
$\epsilon = n^{-\frac{1}{2}}$ for $\cP(\cY)$. Namely, for every $Q \in \cP(\cY)$, there exists a $\vec{k}$ such that $\twonorm{Q - Q_{\vec{k}}} \leq n^{-\frac{1}{2}}$. This can be verified easily since by choosing a $\vec{k}$ that minimizes the distance in all but one direction  (say the last), yielding 
\begin{align*}
  \twonorm{Q - Q_{\vec{k}}}^2 
  &= \sum_{y = 1}^{\abs{\cY} - 1} \big( Q(y) - Q_{\vec{k}}(y) \big)^2 +  \big(Q(|\cY|)-Q_{\vec{k}}(|\cY|) \big)^2\\
&= \sum_{y = 1}^{\abs{\cY} - 1} \big( Q(y) - Q_{\vec{k}}(y) \big)^2+
  \Bigg( \sum_{y = 1}^{\abs{\cY} - 1}  Q_{\vec{k}}(y)-Q(y) \Bigg)^2 \\
  &\leq \sum_{y = 1}^{\abs{\cY} - 1} \bigg( \frac{1}{\sqrt{n \zeta  }} \bigg)^2 +
  \Bigg( \sum_{y = 1}^{\abs{\cY} - 1}  \frac{1}{\sqrt{n  \zeta  }} \Bigg)^2
  = \frac{1}{n} .
\end{align*}
Let us, at this point, provide some intuition for the choice of $Q^{(n)}$ in~\eqref{eqn:Qn}. The first part of the convex combination is used to approximate output distributions induced by inputs types that are close to the set of CAIDs. We choose a weight for each element of the $\epsilon$-net that drops exponentially with the distance from the CAOD. This ensures that the necessary normalization  $F$, does not depend on $n$ even though the number of elements in the net increases with $n$. The smaller weights for types far from the CAIDs will later be compensated by the larger deviation of the corresponding mutual information from the capacity. This is achieved by the second part of the convex combination which we use to match the input types far from the CAIDs.

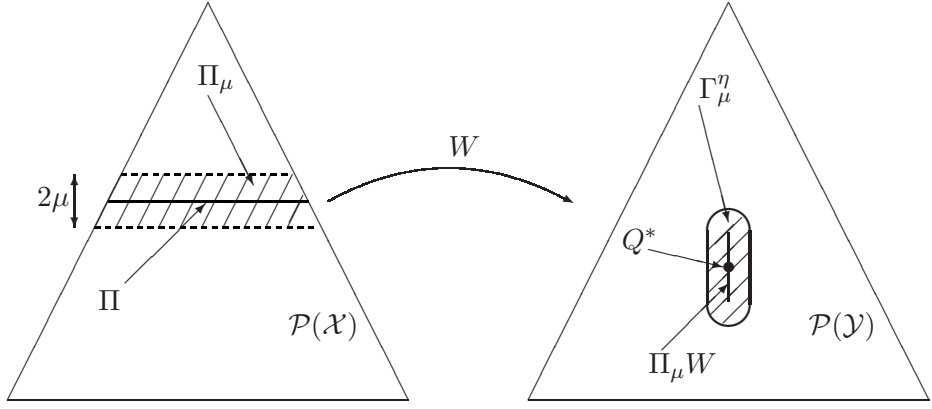
\begin{figure}
\centering
\begin{tabular}{cc}
\begin{picture}(150,150)
\put(0,0){\line(1,0){150}}
\put(0,0){\line(1,2){75}}
\put(75,150){\line(1,-2){75}}
\put(105,25){\mbox{$\cP(\cX)$}}

 \put(34,34){\mbox{$\Pi$}}
 \put(44,44){\vector(1,1){30}}
 \multiput(32.5,65)(8,0){9}{\line(1,2){10}}
 \put(104.5,65){\line(1,2){6}} 
\put(75,115){\vector(1,-2){17}}
 \put(71,119){\mbox{$\Pi_\mu$}}
 \put(25,75){\vector(0,-1){10}}
  \put(25,75){\vector(0,1){10}}
 \put(11,73){\mbox{$2\mu$}}
  \put(165,92){\mbox{$W$}}
 \qbezier(120, 75)(165, 100)(210, 75)
 \put(207,77){\vector(2,-1){5}}
\color{black}
\thicklines
\put(37.5,75){\line(1,0){75}}
 \multiput(42.5,85)(4,0){16}{\line(1,0){2}} 
 \multiput(32.5,65)(4,0){21}{\line(1,0){2}} 

\end{picture} & \hspace{.4in}
\begin{picture}(150,150)
\put(0,0){\line(1,0){150}}
\put(0,0){\line(1,2){75}}
\put(75,150){\line(1,-2){75}}
\put(105,25){\mbox{$\cP(\cY)$}}

  \put(50,19){\vector(1,1){25}}
   \put(45,10){\mbox{$\Pi_\mu W$}}
     \put(64,111){\vector(1,-4){11}}
\multiput(67,40)(0,8){2}{\line(1,1){16}}
\put(68,32){\line(1,1){15}}
\put(67,56){\line(1,1){13}}
\put(73,29){\line(1,1){10}}
\put(64,114){\mbox{$\Gamma_\mu^\eta$}}
\put(75,50){\circle*{4}}
\put(45,57){\vector(4,-1){28}}
\put(35,57){\mbox{$Q^* $}}
\thicklines
\put(75, 37){\line(0,1){26}}
\put(75, 50){\oval(16, 44)}
\end{picture} 
\end{tabular}
\caption{Illustration of the sets in Section~\ref{sec:continuity} for $|\cX|=|\cY|=3$. Here, $\Pi$ is not a singleton and $\Pi_\mu W$ has   measure zero in $\cP(\cY)$ so $W$ is rank-deficient. The unique CAOD $Q^*$ is the image of $\Pi$ under $W$, $\Pi_\mu W$ is the image of $\Pi_\mu$ under $W$ and $\Gamma_\mu^\eta$ is the ``$\eta$-blown-up'' version of $\Pi_\mu W$. }
\label{fig:sets}
\end{figure}

\subsection{Continuity around the CAIDs and the unique CAOD} \label{sec:continuity}

We will often be concerned with probability distributions close to the set of CAIDs $\Pi$ in Euclidean distance, i.e., those distributions belonging to $$\Pi_{\mu} := \Big\{ P \in \cP(\cX)\, \Big|\, \min_{P^* \in \Pi} \twonorm{P - P^*} \leq \mu \Big\}$$ for some small $\mu>0$. 
Sometimes we also need to restrict to probability distributions  in $\Pi_\mu$ with positive conditional information variance. For a constant $v > 0$ we define
$$\Pi_{\mu}^v := \Big\{ P \in\Pi_\mu\,\big|\, V(P,W) \geq v \Big\}.$$

 
The image of $\Pi_{\mu}$ under $W$ is denoted  as $\Pi_{\mu}W$. We also consider a larger, ``$\eta$-blown-up''  version, of $\Pi_\mu W$, namely
$$\Gamma_{\mu}^{\eta} := \Big\{ Q \in \cP(\cY) \, \Big| \, \exists\,  P \in \Pi_{\mu} \textrm{ s.t. } \twonorm{PW - Q} \leq \eta \Big\}.$$
Note that $\Gamma_\mu^0=\Pi_{\mu}W$ if the stochastic matrix $W$ has full rank. See Fig.~\ref{fig:sets} for an illustration. The following Lemma summarizes known results about these sets.

\begin{lemma}
  \label{lm:Pi-mu}
  Let $W:\cX\to\cY$ be a DMC 
  and $v > 0$ be a constant.
  There exists $\mu > 0$ and $\eta > 0$  and  finite constants $V^{+} > 0$, $T^{+} > 0$, $q_{\min}  > 0$, 
  $\alpha > 0$, and $\beta > 0$ such that the following holds.
  For all $P \in \Pi_{\mu}$ and their projections
  $P^* := \argmin_{P' \in \Pi} \twonorm{P - P'}$ and all $Q \in \Gamma_{\mu}^{\eta}$ we have
  \begin{enumerate}
    \item[1.] $Q(y) > q_{\min} $ for all $y \in \cY$,
    \item[2.] $V(W \| Q | P) \geq \frac{V_{\min}}{2}$,
    \item[3.] $I(P, W) \leq C(W) - \alpha  \twonorm{P - P^*}^2$,
    \item[4.] $D(  W \|   Q|P) \leq I(P, W) + \frac{\twonorm{Q - PW}^2}{q_{\min} }$,
    \item[5.] $V(W \| Q | P) \leq V^{+}$ and $T(W \| Q | P) \leq T^{+}$.
  \end{enumerate}
  Furthermore, for any $P \in \Pi_{\mu}^v$ we have
  \begin{enumerate}
    \item[6.] $V(W \| Q | P) \geq \frac{v}{2} > 0$,
    \item[7.] $ \big|\sqrt{V(P,W)} - \sqrt{V(P^*, W)}\, \big| \leq \beta \twonorm{P - P^*}$,  
    \item[8.] $ \big|\sqrt{V(W \| Q | P)} - \sqrt{V(P, W)}\, \big| \leq \beta \twonorm{Q - PW}$.
  \end{enumerate}
\end{lemma}

\begin{proof}
  Properties~1 and~2 hold for small enough $\mu$ and $\eta$ by continuity since $Q^*$ has full support \cite[Cor.\ 1 to Thm.\ 4.5.2]{gallagerIT} and $V(W\|P^*W|P^*) \geq V_{\min}$. The case $V_{\min} = 0$ in Property~2 is
  trivial since $V(W\|Q|P) \geq 0$.  
  Property~3 was established by Strassen~\cite[Eq.\ (4.41)]{Strassen} as well as \PPV~\cite[Eq.\ (501)]{PPV10}. 
  Since $D(  W \|   Q|P) = I(P, W) + D(PW\|Q)$, Property~4 follows immediately from the fact that
  $D(PW\|Q) \leq \frac{1}{\min_{y \in \cY} Q(y)} \twonorm{PW - Q}^2$ (see, e.g.,~\cite[Lem.\ 6.3]{Csi06}).
  Property~5 follows from the fact that $(P, Q) \mapsto V(W\|Q|P)$ and $(P, Q) \mapsto T(W\|Q|P)$ are finite and continuous on
  the compact set $\Pi_{\mu} \times \Gamma_{\mu}^{\eta}$.
   
   Property~6 again holds for small enough $\eta$ by continuity and since $V(W\|P^*W|P) \geq v$ by definition of the set $\Gamma_{\mu}^{\eta}$.
   To verify Properties~7 and~8, note that the quotient $W(y|x)/Q(y)< \infty$ by Property~1. 
   If $W(y|x)/Q(y)=0$, the corresponding terms  in the sums defining $V(P,W)$ and $V(W \|Q |P)$  are  excluded because $\vartheta\log^k \vartheta\to 0$  as $\vartheta \to 0$ for all $k >0$. 
Hence, $P\mapsto V(P,W)$ and $Q \mapsto  V(W \| Q |P)$ are continuously differentiable on $\Pi_\mu$ and $\Gamma_\mu^\eta$ respectively. Because $t\mapsto \sqrt{t}$ is continuously differentiable away from $0$, by Property~6, $P\mapsto  \sqrt{ V(P,W)}$  and  $Q\mapsto \sqrt{V(W \| Q | P)} $  are Lipschitz continuous on $\Pi_\mu$ and $\Gamma_\mu^\eta$ respectively. The uniformity of $\beta$ in $P$ in Property~8  can be verified by explicitly calculating the   derivative of $Q\mapsto \sqrt{V(W \| Q | P)} $ and noting that it can be upper bounded by a finite constant independent of $P$. 
\end{proof}

\subsection{Asymptotics for DMCs} \label{sec:asymp}

We are now ready to prove our main result. Several special cases of~Theorem~\ref{th:main} require additional proof techniques. For the convenience of the reader, we state them separately as propositions. Theorem~\ref{th:main} then follows as a straightforward consequence of these propositions. See Fig.~\ref{fig:flowchart} for a summary. The following proposition considers the ``regular'' case, where the channel and $\eps$ satisfy $V_{\eps} > 0$.

\begin{proposition}
\label{pr:reg}
For every DMC $W$ and $\eps \in (0,1)$ such that $V_{\eps} > 0$, the blocklength $n$, $\eps$-error capacity satisfies
\begin{equation}
\log M^*(W^n,\eps)\le nC+\sqrt{nV_{\eps}} \Phi^{-1}(\eps)+\frac{1}{2}\log n + O(1).  \nonumber
\end{equation}
\end{proposition}
\begin{remark}\label{rm:epshalf}
In the following proof of Proposition~\ref{pr:reg}, we deal with all cases  except $\eps=\frac{1}{2}$,  $V_{\min}=0$ and $V_{\max}=V_\eps>0$.   This special case will be handled in Proposition~\ref{prop:epshalf}(i) as it uses  the  proof techniques in Proposition~\ref{pr:special}. 
\end{remark}
\begin{proof}
Firstly, we employ Proposition~\ref{th:one-shot} to provide a bound on $\log   M^*(W^n, \eps)$. We choose $\delta = n^{-\frac{1}{2}}$, which satisfies $0 < \delta   < 1 - \eps$ for sufficiently large $n$. Substitute the output distribution $Q^{(n)}$ in~\eqref{eqn:Qn} to find
\begin{align*}
\log   M^*(W^n, \eps) \leq  \max_{\vec{x} \in \cX^{\times n}} \underbrace{ D_s^{\eps+\delta}\big( W^n(\cdot|\vec{x}) \big\|Q^{(n)}\big) }_{ =:\ \textrm{cv}(\vec{x})}+ 
    \frac{1}{2} \log n .
\end{align*}
It remains to show that each term $\textrm{cv}(\vec{x})$ in the maximization is upper bounded by $nC + \sqrt{n V_{\eps}} \Phi^{-1}(\eps) + G$ for a suitable constant $G$ 
 for all  sufficiently large $n$. 

We apply Lemma~\ref{lm:Pi-mu}, which supplies us with finite, positive constants $\mu$, $\eta$, $V^+$, $T^+$, $q_{\min} $, $\alpha$ and $\beta$. If $V_{\min} > 0$, we choose $v = \frac{V_{\min}}{2}$ such that $\Pi_{\mu}^v = \Pi_{\mu}$, otherwise $v > 0$ will be specified later. See Case c) below.

We distinguish between three cases for the following; either a) $\vec{x}$ satisfies $P_{\vec{x}} \notin \Pi_{\mu}$ or b) $\vec{x}$ satisfies $P_{\vec{x}}\in \Pi_{\mu}^v$ or c) $\vec{x}$ satisfies $P_{\vec{x}} \in \Pi_{\mu} \setminus \Pi_{\mu}^v$. Note that Case c) is only relevant if $V_{\min} = 0$, as otherwise 
$\Pi_{\mu}^v = \Pi_{\mu}$ by definition of $v$. This strategy in which we  partition   input  types into such classes was proposed by Strassen~\cite[Sec.~4]{Strassen}. See also~\cite[App.~I]{PPV10}. Intuitively, for Case a), $P_{\vec{x}}$ is far from the CAIDs so the first-order term is smaller than capacity; for Case b), $P_{\vec{x}}$  has high conditional information variance and thus bounded skewness so we can apply the Berry-Esseen-type bound of Lemma~\ref{lm:Ds-converse} and; for Case c), $P_{\vec{x}}$  has small conditional information variance so we must use the Chebyshev-type bound  and choose $v$ based on $V_{\max}$ instead of $V_{\min}$.

\subsubsection*{Case a): $P_{\vec{x}}\notin\Pi_\mu$}

The mutual information outside $\Pi_{\mu}$ is bounded away from the capacity, i.e.,   $I(P_{\vec{x}}, W) \leq C' < C$ for all $P_{\vec{x}}\notin\Pi_\mu$. 

Note that $Q^{(n)}$ can be written as a convex combination of the form in Lemma~\ref{lm:convex}, where the index $i$ runs over the sets $\mathcal{K}$ and $\cP_n(\cX)$. 
We first apply Lemma~\ref{lm:convex} to bound $\textrm{cv}(\vec{x})$ with $q(i) = \frac{1}{2\abs{\cP_n(\cX)}}$ and $Q^i = P_{\vec{x}}W^{\times n}$ and then Lemma~\ref{lm:Ds-converse} to bound
\begin{align*}
  \textrm{cv}(\vec{x}) &\leq 
  D_s^{\eps+\delta}\big( W^n(\cdot|\vec{x}) \big\| (P_{\vec{x}}W)^{\times n} \big) 
  + \log \big(2\, \abs{\cP_n(\cX)} \big) \\
  &\leq n I(P_{\vec{x}}, W) + 
    \sqrt{\frac{n V(P_{\vec{x}}, W)}{1-\eps-\delta}} 
    + \log \big(2\, \abs{\cP_n(\cX)} \big) .
    \end{align*}
For the second inequality, we note    that $D_n$ in Lemma~\ref{lm:Ds-converse} evaluates to
$$
D_n=\frac{1}{n}\sum_{i=1}^n \Exp_{W(\cdot|x_i)} \bigg[\log \frac{W(\cdot|x_i)}{ P_{\vec{x}}W (\cdot)}\bigg]=\Exp_{P_{\vec{x}}\times W}\bigg[\log \frac{W }{ P_{\vec{x}}W }\bigg] = D(W \|  P_{\vec{x}}W| P_{\vec{x}})= I(P_{\vec{x}},W),
$$
and similar calculation can be done to show that $V_n=V(P_{\vec{x}},W)$.  Invoking \cite[Lem.\ 62]{PPV10}  and \cite[Rmk.\ 3.1.1]{Han10} yields the uniform bound $V(P_{\vec{x}},W)\le \frac{8\log^2 e}{e^2}\, |\cY|\le 2.3\, |\cY|$. Hence,
\begin{align*}
  \textrm{cv}(\vec{x}) &\leq n C' + \sqrt{n} \sqrt{ \frac{ 2.3 \, |\cY|   }{1-\eps-\delta }}
    +\big( \abs{\cX} - 1 \big)\log \big( n + 1 \big) + \log 2 .
\end{align*}
Since $C' < C$, the linear term dominates the term  growing with the square root of $n$ and the term growing logarithmically in $n$ asymptotically. Hence, it is evident that
$\textrm{cv}(\vec{x}) \leq nC + \sqrt{n V_{\eps}} \Phi^{-1}(\eps)$ for sufficiently large $n$.

\subsubsection*{Case b): $P_{\vec{x}}\in \Pi_\mu^v$}


For each $\vec{x}$, we denote by $Q_{\vec{k}(\vec{x})}$ the element of the $\epsilon$-net  (constructed in Section~\ref{sc:output}) closest to  $P_{\vec{x}}W$. We note that since $\|Q_{\vec{k}(\vec{x})} - P_{\vec{x}}W \|_2\le \epsilon=n^{- \frac{1}{2}}$,  we have $Q_{\vec{k}(\vec{x})} \in \Gamma_{\mu}^{\eta}$ for sufficiently large $n$, which enables us to apply the properties described in Lemma~\ref{lm:Pi-mu} extensively below.
%

We first use Lemma~\ref{lm:convex} with $q(i) = \frac{\exp(-\gamma \twonorm{\vec{k}(\vec{x})}^2)}{2 F}$ and $Q^i = (Q_{\vec{k}(\vec{x})})^{\times n}$ to bound
\begin{align*}
  \textrm{cv}(\vec{x}) \leq
    D_s^{\eps+\delta}\big( W^n(\cdot|\vec{x}) \big\| (Q_{\vec{k}(\vec{x})} )^{\times n} \big) + \gamma \twonorm{\vec{k}(\vec{x})}^2 + \log \big( 2F \big).
\end{align*}
We now employ Lemma~\ref{lm:Ds-converse}, where we choose $P_i = W(\cdot|x_i)$ resulting in $D_n:= D( W \| Q_{\vec{k}(\vec{x})} | P_{\vec{x}})$, $V_n: = V( W \| Q_{\vec{k}(\vec{x})} | P_{\vec{x}}) $  and $T_n := T( W \| Q_{\vec{k}(\vec{x})} | P_{\vec{x}})$.  From Lemma~\ref{lm:Pi-mu}, 
we have that $T_n \leq T^+$ and $0 < \frac{v}{2} < V_n   \leq V^+$. We then introduce the finite constant $B := 1 + 6\sqrt{8}\,T_{+}/v^{\frac{3}{2}}$, while substituting for $\delta = n^{-\frac{1}{2}}$, to find
\begin{align*}
  \textrm{cv}(\vec{x}) &\leq n D( W \| Q_{\vec{k}(\vec{x})} | P_{\vec{x}}) 
  + \sqrt{n V(W \| Q_{\vec{k}(\vec{x})} | P_{\vec{x}})}\, 
  \Phi^{-1}\bigg(\eps + \frac{B}{\sqrt{n}} \bigg) 
  + \gamma \twonorm{\vec{k}(\vec{x})}^2 + \log \big( 2 F \big) .
\end{align*}
We now require that $n \geq N$, where $N$ is chosen large enough such that $\eps + \frac{B}{\sqrt{N}} < 1$. This ensures that the coefficient of the term growing as $\sqrt{n}$ in the above expression is finite.
Next, we use the fact that $\Phi^{-1}$ is infinitely differentiable and 
$V(W \| Q_{\vec{k}(\vec{x})} | P_{\vec{x}}) \leq V_{+}$ is finite to bound
\begin{align*}
  \sqrt{n V(W \| Q_{\vec{k}(\vec{x})} | P_{\vec{x}})} \,\Phi^{-1}\bigg(\eps + \frac{B}{\sqrt{n}} 
  \bigg) \leq \sqrt{n V(W \| Q_{\vec{k}(\vec{x})} | P_{\vec{x}})}\, \Phi^{-1}(\eps) + G_1.
\end{align*}
for some  finite constant $G_1$ and all $n \geq N$. Thus, defining $G_2 := G_1 + \log (2F)$, we find
\begin{align*}
  \textrm{cv}(\vec{x}) &\leq n  D( W \| Q_{\vec{k}(\vec{x})} | P_{\vec{x}}) 
  + \sqrt{n V(W \| Q_{\vec{k}(\vec{x})} | P_{\vec{x}})}\, 
  \Phi^{-1}(\eps) 
  + \gamma \twonorm{\vec{k}(\vec{x})}^2 + G_2,
\end{align*}

Next, we would like to replace $Q_{\vec{k}(\vec{x})}$ with $P_{\vec{x}}W$ in the above bound. This can be done without too much loss due to Lemma~\ref{lm:Pi-mu}, which states that
\begin{align*}
   D( W \| Q_{\vec{k}(\vec{x})} | P_{\vec{x}}) \leq I(P_{\vec{x}}, W) + \frac{\twonormb{P_{\vec{x}}W - Q_{\vec{k}(\vec{x})}}^2}{q_{\min} }
  \leq I(P_{\vec{x}}, W) + \frac{1}{n\, q_{\min} } 
\end{align*}
and
\begin{align*}
  \Big| \sqrt{V(W \| Q_{\vec{k}(\vec{x})} | P_{\vec{x}})} - \sqrt{V(P_{\vec{x}}, W)} \Big|
  \leq \beta \twonormb{P_{\vec{x}}W - Q_{\vec{k}(\vec{x})}} \leq \frac{\beta}{\sqrt{n}} .
\end{align*}
Hence, choosing $G_3 : = \frac{1}{q_{\min} } + \beta \big| \Phi^{-1}(\eps)\big| + G_2$, we find that
\begin{align*}
  \textrm{cv}(\vec{x}) &\leq n I(P_{\vec{x}}, W) + \sqrt{n V(P_{\vec{x}}, W)}\, \Phi^{-1}(\eps) + \gamma \twonorm{\vec{k}(\vec{x})}^2  + G_3 .
\end{align*}

In the following, we use the fact that all distributions (and types) $P_{\vec{x}}$ in $\Pi_{\mu}$ satisfy $I(P_{\vec{x}}, W) \leq  C - \alpha \xi^2$ and $|\sqrt{V(P_{\vec{x}}, W)} - \sqrt{V(P^*,W)}| \leq \beta \xi$, where 
$P^* := \argmin_{P' \in \Pi} \twonorm{P_{\vec{x}} - P'}$  (which is unique) and $\xi := \twonorm{P_{\vec{x}} - P^*}$.
Hence,
\begin{align}
  \textrm{cv}(\vec{x}) &\leq n C + \sqrt{n V(P^*, W)} \Phi^{-1}(\eps) +
     \Big( - \alpha \xi^2 n + \beta \abs{\Phi^{-1}(\eps)} \xi \sqrt{n} + \gamma \twonorm{\vec{k}(\vec{x})}^2 \Big) + G_3   . \label{eq:g3}
\end{align}
It thus remains to show that the term in the bracket is upper bounded by a constant, for an appropriate choice of $\gamma$. Let $\twonorm{W} :=\max\{ \|\vec{u}W\|_2\, |\,\|\vec{u} \|_2 \le 1\}$ be the spectral norm of the matrix $W$. It is easy to see that $\twonorm{W}\le \sqrt{ |\cX|}$. 
From the construction of the $\epsilon$-net in Section~\ref{sc:output}, 
\begin{align*}
  \twonorm{\vec{k}(\vec{x})} &= \sqrt{n \zeta }\, \twonorm{Q_{\vec{k}(\vec{x})} - Q^*} \\
    &\leq \sqrt{n \zeta } \Big( \twonorm{Q_{\vec{k}(\vec{x})} - P_{\vec{x}}W} 
      + \twonorm{P_{\vec{x}}W - Q^*} \Big) \\
    &\leq \sqrt{n \zeta } \bigg( \frac{1}{\sqrt{n}} + \twonorm{W}\, \xi  \bigg) .
\end{align*}
Substituting this bound into~\eqref{eq:g3}, we find that the term in the bracket evaluates to
\begin{align*}
  \big(\gamma \zeta  \twonorm{W}^2 - \alpha \big) \xi^2 n + 
    \big(\beta \abs{\Phi^{-1}(\eps)} + 2 \gamma \zeta  \twonorm{W} \big) \xi \sqrt{n} + \gamma \zeta 
\end{align*}

The expression is a quadratic polynomial in $\xi \sqrt{n}$ and has a finite maximum if we choose 
$\gamma$  such that $\gamma  \zeta  \twonorm{W}^2 < \alpha$. (Note that $\twonorm{W} > 0$ for any channel.) Hence, we can write
\begin{align*}
  \textrm{cv}(\vec{x}) &\leq n C + \sqrt{n V(P^*, W)} \Phi^{-1}(\eps) + G_4
\end{align*} 
for an appropriate constant $G_4$ and $n \geq N$.

\subsubsection*{Case c) $P_{\vec{x}} \in \Pi_{\mu} \setminus \Pi_{\mu}^v$}

Note that this case only appears if $V_{\min} = 0$, $V_{\max} = V_{\eps} > 0$ and $\eps \geq \frac{1}{2}$. 
We consider the case $\eps > \frac{1}{2}$ (cf.\  Remark~\ref{rm:epshalf}) leaving the $\eps=\frac{1}{2}$ case for Proposition~\ref{prop:epshalf}(i).
We have 
\begin{align}
\mathrm{cv}(\vec{x}) &\le D_s^{\eps+\delta}(W^n(\cdot|\vec{x})  \|  (P_{\vec{x}}W)^{\times n}) + \log (2|\cP_n(\cX)|) \nonumber\\
&\le nI(P_{\vec{x}},W)+ \sqrt{ \frac{n V(P_{\vec{x}},W)}{1-\eps-\delta} } + \log (2|\cP_n(\cX)|) \nonumber\\
&\le nI(P_{\vec{x}},W)+ \sqrt{ \frac{n v}{1-\eps-\delta} } + \log (2|\cP_n(\cX)|)  \nonumber
\end{align}
Now we choose $v>0$ to be any constant satisfying 
\begin{equation}  
 \sqrt{\frac{v}{1-\eps-\delta}} + \frac{\log (2|\cP_n(\cX)|)}{\sqrt{n}}\le \sqrt{V_{\max}}\Phi^{-1}(\eps) . \nonumber
\end{equation}
It  is certainly possible to find such a   $v$  since the number of types is polynomial so $\delta$ and the second term on the left are arbitrarily small for large enough $n$. Furthermore, $\sqrt{V_{\max}}\Phi^{-1}(\eps)>0$. This is where $\eps\ne \frac{1}{2}$ is crucial. Uniting the preceding   two bounds yields
\begin{align*}
\mathrm{cv}(\vec{x}) &\le nI(P_{\vec{x}},W)+ \sqrt{nV_{\max}}\Phi^{-1}(\eps) 
\le nC+ \sqrt{nV_{\max}}\Phi^{-1}(\eps) .
\end{align*}

Summarizing the bounds for Cases a), b) and c), we thus have the following asymptotic expansion for all $n$ sufficiently large:
\begin{align*}
  \log M^*(W^n, \eps) &\leq \max_{P^* \in \Pi}  nC + \sqrt{n V(P^*, W)} \Phi^{-1}(\eps) + 
    \frac12 \log n +G_4\\
  &= nC + \sqrt{n V_{\eps}} \Phi^{-1}(\eps) + \frac12 \log n + G_4,
\end{align*}
where the last equality follows by definition of $V_{\eps}$. \end{proof}

Surprisingly, the first-order approximation is accurate up to a constant term if $V_{\eps} = 0$ unless the channel is exotic and $\eps \geq \frac{1}{2}$.

\begin{proposition}
\label{pr:special}
For every DMC $W$ and $\eps \in (0,1)$ such that $V_{\eps} = 0$, the blocklength $n$, $\eps$-error capacity satisfies
$\log M^*(W^n,\eps)\le nC + O(1)$, unless the channel is exotic and $\eps \geq \frac{1}{2}$.
\end{proposition}

\begin{proof}
Again, from our bound on the converse for general channels (Proposition~\ref{th:one-shot}), we have
\begin{equation}
\log M^*(W^n,\eps)\le \max_{\vec{x} \in\cX^{\times n}}\,\,  
\underbrace{D_s^{\eps+\delta} (W^n(\cdot|\vec{x})  \|  Q^{(n)}  )}_{=:\, \textrm{cv}(\vec{x})} + \log \frac{1}{\delta}.
\label{eq:starting-point} 
\end{equation}
We upper bound $\textrm{cv}(\vec{x})$ using Lemma~\ref{lm:convex} (picking out the $\vec{k}=\vec{0}$ term) as follows:
\begin{align*}
\textrm{cv}(\vec{x})\le D_s^{\eps+\delta} \big(W^n(\cdot|\vec{x})  \| ( Q^*)^{\times n} \big)+\log\big(2F\big).
\end{align*}
 We also choose $\delta = \frac{1}{2} - \eps$ if $\eps < \frac{1}{2}$ and $\delta = \frac{1-\eps}{2}$ otherwise; hence, the term $\log \frac{1}{\delta}$ is finite and independent of $n$. Also let $m(\vec{x})$ be the number of non-zero variance letters in $\vec{x}$, i.e., 
$m(\vec{x}):= n  P_{\vec{x}}(\cX_+) = \sum_{i=1}^n 1 \{x_i\in\cX_+ \}$
where $\cX_+:=\{x\in \cX: V(W(\cdot|x) \| Q^*)>0\}$. There exist finite constants $v_{\min}, v_{\max}$ and $t_{\max}$ such that, for every $x\in \cX_+$,
\begin{equation*}
0<v_{\min}\le V(W(\cdot|x) \| Q^*)\le v_{\max}, \qquad \textrm{and} \qquad T(W(\cdot|x) \| Q^*) \leq t_{\max} .
\end{equation*}
By the definitions of $D_n := D(W\|Q^*|P_{\vec{x}})$,  $V_n := V(W\|Q^*|P_{\vec{x}})$ and $T_n := T(W\|Q^*|P_{\vec{x}})$ (cf.~Lemma~\ref{lm:Ds-converse}), we   have
\begin{equation}
\frac{m(\vec{x})}{n} v_{\min}\le V_n\le \frac{m (\vec{x})}{n} v_{\max},\qquad \textrm{and} \qquad
T_n \le \frac{m (\vec{x})}{n} t_{\max} \label{eqn:bounds_VT} .
\end{equation}
Further defining $B_n := 6 \, T_n / V_n^{\frac{3}{2}}$, we thus find
\begin{equation*}
B_n\le\sqrt{ \frac{n}{m( \vec{x})} } L \qquad \textrm{where} \qquad L := \frac{6 \, t_{\max}}{v_{\min}^{3/2}} <\infty.
\end{equation*}

Let $m^*$ be an integer satisfying $L/\sqrt{m^*} \le r'$ where $r'$ is chosen such that $\Phi^{-1}(\frac{1}{2}+r)\le 3\,  r$ for all $r \in [0,r']$. The choice $r'=0.35$ does the job.


For $\eps < \frac{1}{2}$, 
following Strassen's argument~\cite[Eq.\ (4.53)-(4.54)]{Strassen} (see also~\cite[App.~I]{PPV10}), we distinguish between two classes of sequences as follows:
the sequence $\vec{x}$ satisfies either a) $m(\vec{x}) \geq m^*$, or b) $m(\vec{x}) < m^*$.
Finally, c) considers the case where $W$ is not exotic and $\eps \geq \frac{1}{2}$. Intuitively, for Case a), we can use the Berry-Esseen-type bound because $m(\vec{x})$ is large, and hence $B_n$ can be bounded  appropriately; for Case b), we use the Chebyshev-type bound because $m(\vec{x})$ is small and; for Case c), we use the non-exoticness  of $W$ to bound $D_n$ far away from $C$.

\subsubsection*{Case a): $\eps < \frac{1}{2}$ and $m(\vec{x}) \geq m^*$}

We apply the Berry-Esseen-type bound in Lemma~\ref{lm:Ds-converse} to~\eqref{eq:starting-point} to find
\begin{align}
  \textrm{cv}(\vec{x}) &\leq n D_n + \sqrt{n V_n} 
  \Phi^{-1}\left(\eps + \delta + \frac{B_n}{\sqrt{n}}\right) \nonumber\\
  &\leq n D_n + \sqrt{n V_n} \Phi^{-1}\left( \frac{1}{2} + \frac{L}{\sqrt{m(\vec{x})}} \right)
  \leq n D_n + 3 L\, \sqrt{\frac{n V_n}{m(\vec{x})}} \, . \label{eqn:be_type}
\end{align}
Here, we used the fact that $\eps + \delta = \frac{1}{2}$ by definition of $\delta$ and the proof concludes
with the observation that $\frac{n V_n}{m(\vec{x})} \leq v_{\max}$ is bounded by a constant, 
and $D_n \leq C$ for all $\vec{x}$.

\subsubsection*{Case b): $\eps < \frac{1}{2}$ and $m(\vec{x}) < m^*$}

We use the Chebyshev-type bound in Lemma~\ref{lm:Ds-converse} to~\eqref{eq:starting-point}   yielding
\begin{equation}
  \mathrm{cv}(\vec{x}) \leq n D_n + \sqrt{\frac{n V_n}{1-\eps-\delta}} 
  = nD_n + \sqrt{2 n V_n} .\label{eqn:ch_type}
\end{equation}
Since by \eqref{eqn:bounds_VT}, $n V_n 
\le m^* v_{\max}$ and $D_n \leq C$ for all   $\vec{x}$, we find the desired bound.

\subsubsection*{Case c): not exotic, $\eps \geq \frac{1}{2}$}

Lemma~\ref{lm:Ds-converse} applied to~\eqref{eq:starting-point} again yields
\begin{equation*}
 \mathrm{cv}(\vec{x}) \leq n D_n + \sqrt{\frac{nV_n}{1-\eps-\delta} } =
n D_n + \sqrt{\frac{2 n V_n}{1-\eps} } \, ,
\end{equation*}
because in this case, $\delta=\frac{1-\eps}{2}$.   By virtue of the fact that $V_{\max}=0$ and $W$ is not exotic, we have that either
\begin{equation}
D(W(\cdot|x) \| Q^* )< C\quad\mbox{or}\quad V(W(\cdot|x) \| Q^* )=0 \label{eqn:case4_clauses}
\end{equation}
for all symbols $x\in\cX$. If $\cX_{+}$ is empty, we have $V_n = 0$ and the bound is immediate. Otherwise,
we define $\psi := C - \max_{x \in \cX_{+}} D(W(\cdot|x)\|Q^*) > 0$, which is positive due to the condition in~\eqref{eqn:case4_clauses}.

Using this, we find that $n D_n \leq n C - m(\vec{x}) \psi$ 
and $n V_n \leq v_{\max} m(\vec{x})$ by~\eqref{eqn:bounds_VT}. Thus,
\begin{equation*}
  \mathrm{cv}(\vec{x}) \leq n C - m(\vec{x}) \psi + \sqrt{\frac{ 2 m(\vec{x}) v_{\max}}{1-\eps} }
\end{equation*}
The latter two terms constitute a quadratic polynomial in $\sqrt{m(\vec{x})}$, and hence, their sum has a finite maximum. 
\end{proof}

Finally, we deal with the case that was left out in Proposition~\ref{pr:reg}. 
\begin{proposition}
\label{prop:epshalf}
Let $\eps=\frac{1}{2}$. The following hold:
\begin{itemize}
\item[(i)] For every DMC $W$ such that $V_{\min}=0$ and  $V_{\max}>0$, the blocklength $n$, $\eps$-error capacity satisfies $\log M^*(W^n,\eps)\le nC + \frac{1}{2}\log n + O(1)$. 
\item[(ii)]  For every exotic DMC $W$ (in particular, $V_{\max}=0$),   the same bound as in (i) holds.
\end{itemize}
\end{proposition}
\begin{proof}
By placing no assumptions on $V_{\max}\ge 0$, we can prove both parts in tandem. The proof follows  closely that of Proposition~\ref{pr:special} with the exception that we choose $\delta=n^{-\frac{1}{2}}$ so the $\log \frac{1}{\delta}$ term evaluates to $\frac{1}{2}\log n$. It remains to show that $\textrm{cv}(\vec{x})\le nC + O(1)$.  We split the analysis into Cases a) and b) as in  Proposition~\ref{pr:special} and let $D_n := D(W\|Q^*|P_{\vec{x}})$ and $V_n := V(W\|Q^*|P_{\vec{x}})$. 
\subsubsection*{Case a): $\eps=\frac{1}{2}$, $V_{\min}=0$   and $m(\vec{x})\ge m^*$ }
By the same  steps that led to \eqref{eqn:be_type},  we have
\begin{equation}
  \textrm{cv}(\vec{x})\le n D_n + 3 \,  (L+1)\, \sqrt{\frac{n V_n}{m(\vec{x})}}\, \nonumber
\end{equation}
because $\delta=n^{-\frac{1}{2}}$. We obtain the desired bound by noting that $\frac{n V_n}{m(\vec{x})}\le v_{\max}$ and $D_n\le C$. 
\subsubsection*{Case b): $\eps=\frac{1}{2}$, $V_{\min}=0$  and $m(\vec{x}) <  m^*$ }
By the same steps that led to \eqref{eqn:ch_type},  we have
\begin{equation}
  \textrm{cv}(\vec{x})\le n D_n + \sqrt{4 n V_n} \nonumber
\end{equation}
because $1-\eps-\delta=\frac{1}{2}-\delta\ge \frac{1}{4}$ for all   $n\ge 4$. The proof is completed by noting that $nV_n\le m^*v_{\max}$ and $D_n\le C$.\end{proof} 


\begin{proof}[Proof of Theorem~\ref{th:main}]
The first statement follows by Propositions~\ref{pr:reg} and \ref{prop:epshalf}(i). The second statement follows by Proposition~\ref{pr:special}.
\end{proof}

\section{Conclusion and Open Problems}

We have presented improved converse (upper) bounds on the blocklength $n$, $\eps$-average error capacity $M^*(W^n,\eps)$. These bounds are tight in the third-order for all DMCs with positive reverse dispersion~\cite[Thm.~53]{Pol10}. However, the BEC (with zero reverse dispersion) is a notable example for which our result is not tight and in fact overestimates  $\log M^*(W^n,\eps)$ by $\frac12 \log n$. To prove a tight converse bound on the third-order for the BEC, a different non-product choice for $Q^{(n)}$ is necessary, as was pointed out recently by Polyanskiy~\cite[Thm.~23]{Polyanskiy13}. It remains to investigate whether a combination of Polyanskiy's choice and our choice of output distribution can be used to derive tight third-order asymptotic  bounds for \emph{all} DMCs.

Our general converse bound in Proposition~\ref{th:one-shot} can be specialized to channels with cost constraints. As such, it can be applied to the AWGN channel with maximal (or equal) power constraints and the evaluation of Proposition~\ref{th:one-shot} using the product CAOD yields the $\frac12 \log n + O(1)$ upper bound on the third-order term~\cite[Thm.~54]{PPV10}.  It would be interesting to check if the evaluation of Proposition~\ref{th:one-shot} yields the same upper bound for the finite-dimensional infinite constellations problem~\cite[Thm.~13]{ingber13}. 

\subsubsection*{Acknowledgements}

MT thanks Ligong Wang for helpful explanations. VYFT thanks Yanina Shkel for insightful discussions and Pierre Moulin for sharing his ITA paper~\cite{Mou12a}. 
MT is supported by the National Research Foundation and the Ministry of Education of Singapore.
VYFT would like to acknowledge funding support from the Agency for Science, Technology and Research (A*STAR), Singapore.

\bibliographystyle{unsrt}

\bibliography{isitbib}

\end{document}